\documentclass[conference]{IEEEtran}

\IEEEoverridecommandlockouts
\usepackage{cite}
\usepackage{textcomp}
\def\BibTeX{{\rm B\kern-.05em{\sc i\kern-.025em b}\kern-.08em
    T\kern-.1667em\lower.7ex\hbox{E}\kern-.125emX}}

\usepackage{comment}
\usepackage{amsthm}
\usepackage{amsmath,amssymb,amsfonts}
\usepackage{graphicx}
\usepackage{subfig}
\usepackage{algpseudocode}
\usepackage{algorithm}
\usepackage{lipsum}

\makeatletter
\newcommand\fs@betterruled{%
  \def\@fs@cfont{\bfseries}\let\@fs@capt\floatc@ruled
  \def\@fs@pre{\vspace*{7pt}\hrule height.8pt depth0pt \kern2pt}%
  \def\@fs@post{\kern2pt\hrule\relax}%
  \def\@fs@mid{\kern2pt\hrule\kern2pt}%
  \let\@fs@iftopcapt\iftrue}
\floatstyle{betterruled}
\restylefloat{algorithm}
\makeatother

\newtheorem{theorem}{Theorem}[section]

\newtheorem{lemma}[theorem]{Lemma}

\newcommand{\paren}[1]{\left( #1 \right)}
\newcommand{\cbrace}[1]{\left\{#1\right\}}
\newcommand{\sbrace}[1]{\left[#1\right]}

\usepackage{mathtools}

\newcommand{\cC}{{\cal C}}

\newcommand{\cS}{{\cal S}}

\newcommand{\vc}{\vec{c}}

\newcommand{\ve}{\vec{e}}
\newcommand{\vi}{\vec{i}}
\newcommand{\vs}{\vec{s}}
\newcommand{\vu}{\vec{u}}
\newcommand{\vv}{\vec{v}}
\newcommand{\vw}{\vec{w}}
\newcommand{\vx}{\vec{x}}
\newcommand{\vy}{\vec{y}}
\newcommand{\vz}{\vec{z}}

\newcommand{\vchat}{\vc_{*}}
\newcommand{\vzero}{\vec{0}}
\newcommand{\encode}{\alpha}

\newcommand{\gf}{\mathbb{F}}
\newcommand{\gfq}{\gf_{q}}
\newcommand{\gftwo}{\gf_{2}}

\newcommand{\complex}{\mathbb{C}}

\newcommand{\ym}{{{\demod{y}}}}
\newcommand{\vym}{\demod{{\vy}}}

\newcommand{\demodsym}{\theta}
\newcommand{\demod}[1]{\demodsym\paren{#1}}
\newcommand{\govern}[1]{\gamma\paren{{#1}}}

\newcommand{\bigo}{\mathcal{O}}

\usepackage{xcolor}
\newcommand{\asol}[1]{{\color{blue}{#1}}}

\newcommand{\ebno}{E_{b}/N_{0}}
\newcommand{\db}{\sbrace{dB}}
\newcommand{\crclen}{\text{crc}_\text{len}}
\newcommand{\SNR}{\text{SNR}}

\newcommand{\jstar}{j_{*}}
\newcommand{\predsym}{\pi}
\newcommand{\pred}[1]{\predsym\paren{#1}}
\newcommand{\pve}{q\paren{\ve}}
\newcommand{\pvx}{q\paren{\vx}}
\newcommand{\ppredve}{q\paren{\pred{\ve}}}
\newcommand{\prev}[1]{\phi\paren{#1}}

\begin{document}

\title{Soft Maximum Likelihood Decoding using GRAND}
%
%
%

\author{\IEEEauthorblockN{Amit Solomon}
\IEEEauthorblockA{\textit{RLE, MIT} \\
Cambridge, MA 02139, USA \\
amitsol@mit.edu}
\and
\IEEEauthorblockN{Ken R. Duffy}
\IEEEauthorblockA{\textit{Hamilton Institute} \\
{Maynooth University, Ireland}\\
ken.duffy@mu.ie}
\and
\IEEEauthorblockN{Muriel M\'edard}
\IEEEauthorblockA{\textit{RLE, MIT} \\
Cambridge, MA 02139, USA \\
medard@mit.edu}
}

\markboth{Journal of \LaTeX\ Class Files,~Vol.~14, No.~8, August~2015}%
{Shell \MakeLowercase{\textit{et al.}}: Bare Demo of IEEEtran.cls for IEEE Journals}

\maketitle

\begin{abstract}
Maximum Likelihood (ML) decoding of forward error correction codes is known to be optimally accurate, but is not used in practice as it proves too challenging to efficiently implement. Here we introduce a ML decoder called SGRAND, which is a development of a previously described hard detection ML decoder called GRAND, that fully avails of soft detection information and is suitable for use with any arbitrary high-rate, short-length block code. We assess SGRAND's performance on CRC-aided Polar (CA-Polar) codes, which will be used for all control channel communication in 5G NR, comparing its accuracy with CRC-Aided Successive Cancellation List decoding (CA-SCL), a state-of-the-art soft-information decoder specific to CA-Polar codes.
\end{abstract}

\IEEEpeerreviewmaketitle

\section{Introduction}
Since the work of Shannon~\cite{shannon1948mathematical}, Maximum likelihood (ML) decoders have been sought. As it was established in the 1970s that ML decoding of arbitrary linear codes is an NP-complete problem~\cite{berlekamp1978inherent}, instead of seeking a universal, code book independent decoder, most codes are co-designed and developed with a specific decoder that is often an approximation of a ML decoder~\cite{roth2006introduction,lin2001error}. 

An exception to this paradigm is the recently introduced Guessing Random Additive Noise Decoding (GRAND)~\cite{duffy2018guessing}. In the hard detection setting with additive noise, GRAND has been formally proven to be a ML decoder that works with any block code construction. Its novelty stems from the fact that it attempts to identify the noise that corrupted the code word, rather than directly determining the transmitted code word itself. A computationally more efficient variant of GRAND is GRAND with ABandonment (GRANDAB), which either finds a ML code word or reports a failure after a pre-determined computational cut-off. A failure report from  GRANDAB is equivalent to reporting a channel erasure, which is preferable as erasures require less overhead when corrected by an outer code. Both GRAND and GRANDAB have been mathematically established to be capacity achieving when used with random codes~\cite{duffy2019capacity}.

Incorporating soft information into GRAND and GRANDAB decoding would improve accuracy, but it is unclear how to fully do so. An initial attempt that uses one bit of quantized soft information per channel use, which is similar in spirit to how soft information is generated for use with Chase decoding \cite{Cha72}, has been proposed~\cite{duffy2019guessing}. That results in improved accuracy and reduced decoding complexity, but falls short on fully making use of soft information.

In this paper we introduce an advanced variant of GRAND called Soft GRAND (SGRAND) that fully utilizes real-valued channel outputs. We prove that SGRAND is a ML decoder for arbitrary additive memoryless channels, and benchmark its performance by comparison with the state-of-the-art soft detection decoder of 5G NR CA-Polar codes, CRC-Aided Successive Cancellation List decoding (CA-SCL)~\cite{tal2011list,tal2015list}, as implemented in the Matlab 5G toolbox. Fig.~\ref{fig:bler} provides a representative performance evaluation for a [128,105] CA-Polar code where it can be seen that SGRAND with ABandonment (SGRANDAB) obtains better Block Error Rate (BLER), which includes both erroneous decodings and erasures, than CA-SCL. The gap between the two curves at a BLER of $10^{-3}$ is $\approx 0.5$ dB.

The paper is structured as follows. In Section~\ref{sec:model} we discuss the communication model and review relevant background. In Section~\ref{sec:sgrand_m_1} we present SGRAND. In Section~\ref{sec:simulations} we explain how the simulation was conducted and show additional results. We conclude the paper and discuss future work in Section~\ref{sec:conclusion}.

\begin{figure}
    \centering
    \includegraphics[width=1 \columnwidth]{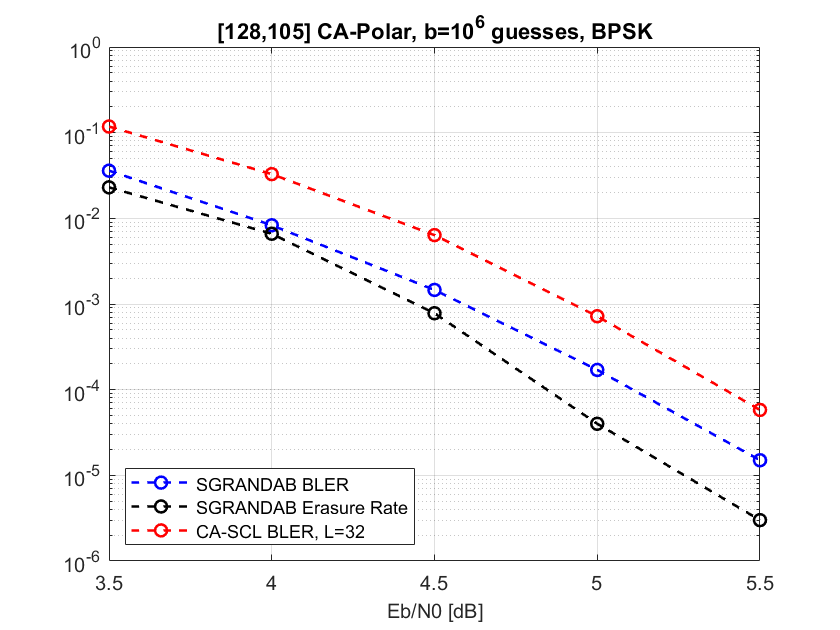}
    \caption{SGRANDAB and CA-SCL BLER comparison for a [128,105] CA-Polar code. CA-SCL takes list size, $L$, as an argument. While MATLAB defaults to $L=8$, better BLER performance is seen for higher values and so results for $L=32$ are presented. Also shown is SGRANDAB's errasure rate. }
    \label{fig:bler}
\end{figure}

\section{Model and Background}\label{sec:model}

\subsection{Definitions and Notation}
Let $x,\vec{x},X, \vec{X}$ denote a scalar, vector, matrix, and a random vector respectively. All vectors are row vectors, and $x_i$ denotes the $i$-th element of $\vx$. A linear block code is characterized by a code-length, $n$, and code-dimension, $k$, $\sbrace{n,k}$. $\gfq$ denotes a Galois field with $q$ elements. Composition of functions is denoted by $\circ$, i.e. $f\circ g\paren{x}=f\paren{g\paren{x}}$.

\subsection{System model}
We consider the standard binary communication setting where an information word $\vu\in\gftwo^k$ is encoded with an error correcting code $\encode:\gftwo^k\to\gftwo^n$ into a code word $\vc\in\gftwo^n$. The code book is the set of all code words $\cC=\cbrace{\vc:\vc=\encode\paren{\vu},\vu\in\gftwo^k}$. For communication on an analog medium, each set of $m$ bits is modulated into a channel input $\vx\in\complex^{n/m}$, using a modulation function $\gftwo^n\to\complex^{n/m}$, that is in turn sent over a noisy channel, where, for convenience, we assume that $m$ divides $n$. The channel output is denoted by $\vec{Y}=\vx+\vec{Z}\in\complex^{n/m}$, where $\vec{Z}$ is random additive noise. A soft decoder is a function $\complex^{n/m}\to\cC$ that outputs an estimate of $\vc$ from $\vy$. A soft detection ML decoder outputs $\underset{\vc\in\cC}{\arg\max}\;p\paren{\vy\mid\vc}$

Each channel output governs the a posteriori probabilities of $m$ demodulated bits, which, as a result of interleaving, are not necessarily sequential. As noise is assumed to impact each transmitted symbol independently, the a posteriori probability of each demodulated bit's value depends on a single channel output. We denote the index of the channel output that determines the a posteriori probability of the $i$-th bit by $\govern{i}$.

A demodulator is a function $\demodsym:\complex^{n/m}\to\gftwo^n$ that returns the most likely modulated version of the channel output, returning $\vym=\underset{\vw\in\gftwo^n}{\arg\max}\;p\paren{\vy\mid\vw}$. A hard decoder is a function $\gftwo^n\to\cC$ that receives $\vym\in\gftwo^n$ and returns a code word.

\subsection{GRAND}
GRAND is a hard decoder that identifies the most likely noise sequence, $\vz\in\gftwo^n$ as GRAND is a hard decoder, rather than the most likely code word $\vc$. It was shown in~\cite{duffy2018guessing,duffy2019capacity} that these two problems are equivalent for not-necessarily-memoryless discrete additive channels with uniformly likely input. Identification is achieved by inverting, in order from most likely to least likely based on a probabilistic noise model, the effect of each putative noise sequence from the hard demodulated sequence and querying whether what remains, $\vym-\vz$, is in the code book. GRAND only stops at the first instance where the response is affirmative, while GRANDAB additionally halts and reports an erasure if too many code book queries have been made. 

Linear codes are the most common forward error correction codes~\cite{roth2006introduction},\cite{lin2001error}.
Testing for code book membership in a linear code book is equivalent to determining if the syndrome of $\vym-\vz$ is the zero vector, that is if $H\paren{\vym-\vz}^T=\vzero^T$, where $H$ is a parity check matrix of the code~\cite{roth2006introduction},\cite{lin2001error}. Pseudocode for GRANDAB and linear codes is given in Algorithm~\ref{alg:GRAND}.

\begin{algorithm}
\caption{GRANDAB for a linear code}
\label{alg:GRAND}
\begin{flushleft}
        \textbf{Input:} $\vym,H,b$ \Comment{$b$ is max \#queries, GRAND: $b=\infty$}\newline
        \textbf{Output:} $\vchat$
\end{flushleft}
\begin{algorithmic}[1]
\State $g\gets0$ \Comment{$g$ counts queries performed}
\While{$g\leq b$}
\State $\vz\gets$ next most likely noise putative sequence\label{algline:next}
\State $g\gets g+1$
\If{$H\cdot\paren{\vym-\vz}^T=\vzero$} \Comment{$\gftwo$ operations}
\State $\vchat\gets\vym-\vz$ \Comment{ML code word found}
\State \Return
\EndIf
\EndWhile
\State $\vchat\gets\perp$ \Comment{Code word not found in $b$ queries; erasure}
\State \Return 
\end{algorithmic}
\end{algorithm}

Note that Algorithm~\ref{alg:GRAND} does not specify how to perform the choice of the next putative noise sequence to be considered in line~\ref{algline:next}. For a Binary Symmetric Channel (BSC), that amounts to generating all sequences of Hamming weight 0, followed by all sequences of Hamming weight 1 and so on, which is a well-studied problem for which efficient recursive algorithms are known \cite{ruskey2009coolest,mazurkiewicz2017efficient}. Modelling a well-interleaved communication system, in this paper we establish how to perform GRAND with per-bit soft information for an arbitrary binary memoryless channel. 

\section{SGRAND}\label{sec:sgrand_m_1}
We capture the rank-ordering of demodulated symbols by their reliability via an Ordered Error Indices (OEI) vector $\vi=\paren{i_1,\ldots,i_{n}}$, which satisfies $p\paren{\vym_{i_j}\mid {y_{\govern{i_j}}}}\leq p\paren{\vym_{i_l}\mid {y_{\govern{i_l}}}}$ for $j<l$, where $\vym$ is the most likely demodulated version of $\vy$, and the OEI vector depends on $\vy$, the modulation and the channel statistics, but this dependency is suppressed for notational simplicity. We assume that a priori demodulated probabilities are equal, corresponding to uniform input. The OEI vector is almost surely unique for a continuous memoryless noise distribution.

Pseudocode for SGRANDAB is given in Algorithm~\ref{alg:sgrandab}, which is most readily understood by the example that follows. Suppose $n=3,m=1$, and the channel output realizations satisfy:
$p\paren{\vym_1\mid y_{\govern{1}}}=0.6$, $p\paren{\vym_2\mid y_{\govern{2}}}=0.8$,  $p\paren{\vym_3\mid y_{\govern{3}}}=0.9$, and the probability density function values are $p\paren{y_{\govern{1}}}=0.5$, $p\paren{y_{\govern{2}}}=1$, $p\paren{y_{\govern{3}}}=4$. The unique OEI vector is $\paren{1,2,3}$. For the sake of the example, we assume that $H\cdot\vchat^T=\vzero^T$ is satisfied only at the last query, identifying an element of the code book. The workings of Algorithm~\ref{alg:sgrandab} for this setting are given in Table~\ref{tbl:guess_example}. We have the following theorem on correctness and progress for SGRAND.

\begin{algorithm}
\caption{SGRANDAB for a linear code}
\label{alg:sgrandab}
\begin{flushleft}
        \textbf{Input:} $\vy,H,b$ \Comment{$b$ is max \#queries, SGRAND: $b=\infty$}\newline
        \textbf{Output:} $\vchat$
\end{flushleft}
\begin{algorithmic}[1]
\State $g\gets0$ \Comment{$g$ counts queries performed}
\State $\cS\gets\cbrace{\vzero}$ \Comment{$\cS$ contains candidate error vectors}
\State $\vi\gets$ OEI vector \Comment{Based on $\vy$}
\While{$g\leq b$}
\State $\ve\gets\underset{\vv\in\cS}{\arg\max}\;p\paren{\vy\mid \vym - \vv}$
\State $\cS = \cS\setminus\cbrace{\ve}$
\State $g\gets g+1$
\If{$H\cdot\paren{\vym-\ve}^T=\vzero$} \label{algline:synd_check}
\State $\vchat\gets\vym-\ve$ \Comment{ML code word found}
\State \Return
\Else \Comment{Update $\cS$ for the next query}
\If{$\ve=\vzero$}
\State $\jstar\gets 0$
\Else
\State $\jstar\gets\max\cbrace{j:{e_i}_j\neq 0}$ \Comment{$\jstar>0$}
\EndIf
\If{$\jstar<n$}
\State ${e_i}_{\paren{\jstar+1}}\gets 1$
\State $\cS=\cS\cup\cbrace{\ve}$
\If{$\jstar>0$}
\State $e_{i_{\jstar}}\gets 0$
\State $\cS=\cS\cup\cbrace{\ve}$
\EndIf
\EndIf
\EndIf
\EndWhile
\State $\vchat\gets\perp$ \Comment{Code word not found in $b$ queries; erasure}
\State \Return 
\end{algorithmic}
\end{algorithm}

\begin{table}[t]
\centering
\begin{tabular}{|c||c||c||c||c|}
\hline
$g$ & $\ve$ & $p$ & $\jstar$ & $\cS$\\
\hline
$1$ & $\paren{0,0,0}$ & $0.432$ & $0$ & $\cbrace{\paren{1,0,0}}$\\
\hline
$2$ & $\paren{1,0,0}$ & $0.288$ & $1$ & $\cbrace{\paren{1,1,0},\paren{0,1,0}}$\\
\hline
$3$ & $\paren{0,1,0}$ & $0.108$ & $2$ & $\cbrace{\paren{1,1,0},\paren{0,1,1},\paren{0,0,1}}$\\
\hline
$4$ & $\paren{1,1,0}$ & $0.072$ & $2$ & $\cbrace{\paren{0,1,1},\paren{0,0,1},\paren{1,1,1},\paren{1,0,1}}$\\
\hline
$5$ & $\paren{0,0,1}$ & $0.048$ & $3$ & $\cbrace{\paren{0,1,1},\paren{1,1,1},\paren{1,0,1}}$\\
\hline
$6$ & $\paren{1,0,1}$ & $0.032$ & $3$ & $\cbrace{\paren{0,1,1},\paren{1,1,1}}$\\
\hline
$7$ & $\paren{0,1,1}$ & $0.012$ & $3$ & $\cbrace{\paren{1,1,1}}$\\
\hline
$8$ & $\paren{1,1,1}$ & $0.008$ & $3$ & $\emptyset$\\
\hline
\end{tabular}
\caption{An example of the execution of Algorithm~\ref{alg:sgrandab}. Column $\ve$ gives the last queried noise sequence. Column $\cS$ is the state of the set after each query has been made, at the end of the while loop when all new putative error sequences have been added. The column marked $p$ reports $p\paren{\vy\mid\vym-\ve}$.}
\label{tbl:guess_example}
\end{table}

\begin{theorem}\label{thm:sgrand_correctness}
For an additive memoryless channel, Algorithm~\ref{alg:sgrandab} satisfies the following properties:
\begin{samepage}
\begin{enumerate}
    \item \textbf{Correctness:} Error vectors are queried in non-increasing order of likelihood, therefore a returned code word is a ML code word. \label{sgrand_prop_correctness}
    \item \textbf{Progress:} Each error vector is queried at most once.\label{sgrand_prop_progress}
\end{enumerate}
\end{samepage}
\end{theorem}

\subsection{Proof of Theorem~\ref{thm:sgrand_correctness}}\label{sub:proof_thm}
\begin{proof}
For an OEI vector $\vi$ we define the \textit{parent} of $\ve\neq\vzero$, denoted by $\pred{\ve}$, as the following vector:
  \begin{align*}
    {\paren{\pred{\ve}}}_{i_j} = \left\{\begin{array}{lr}
        0, & \text{if } j=\jstar, \\
        1, & \text{if } j=\jstar-1\\
        e_{i_j}, & \text{otherwise}
        \end{array}\right.
  \end{align*}
where $\jstar$ is defined as in Algorithm~\ref{alg:sgrandab} and $\pred{\vzero}$ is not defined. We say that $\ve$ is a \textit{child} of $\pred{\ve}$. While a parent is unique, most, but not all, vectors have two children. The only vectors that do not have two children are:
$\vzero$, which has a unique child, the vector has $1$ in $e_{i_1}$ and zero elsewhere; and those where $\jstar=n$, which have no children. For instance, for an OEI vector $\vi=\paren{1,2,3}$,  $\paren{1,1,0}=\pred{\paren{1,1,1}}=\pred{\paren{1,0,1}}$. 
To establish the theorem, first we prove the following lemma.
\begin{lemma}
The following are true for all $\ve\neq\vzero$ and an arbitrary memoryless channel:
\begin{enumerate}\label{lem:predecessor_properties}
    \item $\pred{\ve}\neq\ve$
    \item $p\paren{\vy\mid\vym-\ve}\leq p\paren{\vy\mid\vym-\pred{\ve}}$\label{lemprop:pred_prob}
    \item $\predsym\circ\ldots\circ\pred{\ve}=\vzero$ after $\jstar$ compositions.\label{lemprop:pred_chain}
\end{enumerate}
\end{lemma}

\begin{proof}
Let $\ve\neq\vzero$, $\vi$ be an OEI vector, and $\jstar$ (with respect to $\ve$) be defined as in Algorithm~\ref{alg:sgrandab}.

\noindent 1) For any such $\ve$, $\paren{\ve}_{i_{\jstar}}=1, \paren{\pred{\ve}}_{i_{\jstar}}=0$, so $\ve\neq\pred{\ve}$.

\noindent 2) Let $q_j\paren{x}=p\paren{y_{\govern{i_j}}\mid\vym_{i_j}-x}$, $\pvx=\prod_{j=1}^{n}q_j\paren{x_{i_j}}$. The existence of $j$ such that $j:q_j\paren{e_{i_j}}=0$ completes the proof. Otherwise assume $q\paren{\ve}>0$. If $\jstar=1$ or $e_{i_{\jstar-1}}=1$, then $\pve=q_{\jstar}\paren{1}\prod_{j\neq\jstar}q_{j}\paren{e_{i_j}} $, and $\ppredve=q_{\jstar}\paren{0}\prod_{j\neq\jstar}q_{j}\paren{e_{i_j}}$, which completes the proof, as $\forall j: q_{j}\paren{1}\leq q_{j}\paren{0}$. Otherwise assume that $\jstar>1,e_{\jstar-1}=0$. Then $\pve=q_{\jstar-1}\paren{0} q_{\jstar}\paren{1} \prod_{j\neq\jstar-1,\jstar}q_{j}\paren{e_{i_j}} $, and $\ppredve=q_{\jstar-1}\paren{1}\cdot q_{\jstar}\paren{0}\cdot\prod_{j\neq\jstar}q_{j}\paren{e_{i_j}}$. Note that $\ppredve/\pve=q_{\jstar-1}\paren{1} q_{\jstar}\paren{0}/(q_{\jstar-1}\paren{0} q_{\jstar}\paren{1})\geq 1$ which completes the proof as $\jstar$ comes from an OEI vector.

\noindent 3) Let $\ve\neq 0, \jstar=\max\cbrace{j:{e_i}_j\neq 0}$. By definition of $\pred{\ve}$, we get $\jstar-1 = \max\cbrace{j:{\paren{\pred{\ve}}_i}_j\neq 0}$. By iterating this argument $\jstar$ times, we get the zero vector.
\end{proof}

Note that after each query Algorithm~\ref{alg:sgrandab} adds to $\cS$ all the children of the queried noise sequence $\ve$, after removing $\ve$ from $\cS$. We return to the proof Theorem~\ref{thm:sgrand_correctness} using Lemma~\ref{lem:predecessor_properties}. Observe that in order to prove Property~\ref{sgrand_prop_progress}, it is sufficient to prove that an error vector is added to $\cS$ at most once.

We prove properties~\ref{sgrand_prop_correctness} and \ref{sgrand_prop_progress} by induction on the number of queries $g$, evaluated at the end of the while loop. The case of $g=1$ follows immediately as $\cS$ is initialized to contain only $\ve=\vzero$, which is the most likely error. Then the newly added vector is the unique child of $\vzero$, which is not $\vzero$, so Property~\ref{sgrand_prop_progress} is satisfied. We now assume that the properties are satisfied after $g'$ queries, and establish that they are satisfied after $g'+1$ queries. For Property~\ref{sgrand_prop_progress} suppose by contradiction that after $g'+1$ queries, a vector $\vv$, that was previously added to $\cS$, is added to $\cS$. A vector is added to $\cS$ only when its parent is queried, so $\vv\neq\vzero$, as $\vzero$ has no parent. We conclude that at the $g'+1$ query, the unique parent $\ve=\pred{\vv}$ is queried, which means that $\ve$ has been added more than once within $g'$ queries, which contradicts the induction assumption. 

For Property~\ref{sgrand_prop_correctness} suppose by contradiction that it is not satisfied, and the next most likely error vector $\vv\not\in\cS$ that has not been queried before is more likely than $\ve$, i.e. $p\paren{\vy\mid\vym-\ve}<p\paren{\vy\mid\vym-\vv}$, and therefore $\vv$ should have been queried at query number $g'+1$. We know that $\vv\neq\vzero$, as $\vzero$ is always the first queried error. $\pred{\vv}\in\cS$ cannot hold, as $\pred{\vv}$ is at least as likely as $\vv$ due to Lemma~\ref{lem:predecessor_properties}, which would contradict $p\paren{\vy\mid\vym-\ve}<p\paren{\vy\mid\vym-\vv}$. Furthermore, if $\pred{\vv}\in\cS$ ever held, then $\vv$ would be added to $\cS$, which contradicts the assumption that it was never added to $\cS$. Since $\pred{\vv}\not\in\cS$, and has never been, it means that $\pred{\vv}$ has not been previously queried. If $p\paren{\vy\mid\vym-\vv}<p\paren{\vy\mid\vym-\pred{\vv}}$, we contradict the assumption that $\vv$ should have been the next query. Otherwise, assume $p\paren{\vy\mid\vym-\vv}=p\paren{\vy\mid\vym-\pred{\vv}}$ and both $\vv,\pred{\vv}$ have not been previously queried. By repeating this argument $\jstar$ times, we get that $\vzero$ has not been queried (see Lemma~\ref{lem:predecessor_properties}), which contradicts the fact that it is always the first query of the algorithm.
\end{proof}

\subsection{Complexity considerations}\label{sub:complexity}
After each iteration of the algorithm, one vector is removed from $\cS$, and at most two are added to it. Therefore $\cS$ grows by at most one after each query, so after $g$ queries it contains at most $g+1$ elements. One efficient way of implementing $\cS$ is via a Max-Heap~\cite{cormen2009introduction}. At each iteration Algorithm~\ref{alg:sgrandab} performs the following: extracts the most likely error vector of $\cS$, performs matrix multiplication, and adds at most two elements to $\cS$. The complexity of the matrix multiplication, denoted $f\paren{n,k}$, differs across implementations, e.g.~\cite{burgisser2013algebraic,frigo1999cache}. Hence the complexity of the algorithm after $g$ queries, when $\cS$ is implemented with a Max-Heap, is $\bigo\paren{g\cdot f\paren{n,k}\cdot\log g}$. As a result, the worst-case complexity of SGRANDAB with an abandonment threshold of $b$ (as defined in Algorithm~\ref{alg:sgrandab}) is $\bigo\paren{b\cdot f\paren{n,k}\cdot\log b}$. This is a worst-case analysis, and in practice SGRANDAB often finds a code word much earlier before reaching the abandonment threshold, as can be seen in the simulated results. We also note that the algorithm can be readily parallelized, reducing its time complexity. One  mechanism to do so would be to check membership of the code book using different vectors in $\cS$ concurrently, while keeping track of each error vector's respective likelihood.

\subsection{SGRAND, no intra-symbol interleaving}
In almost all communication systems, and at many layers, bits are interleaved so that correlations in noise that corrupts adjacent bits does not result in correlation at the level of information bits. Interleaving is often performed on bits within the same modulated symbol, as well as across bits from different modulated symbols. While most decoders are designed under assumptions of substantial interleaving, we show how SGRAND can be altered to operate when bits within a modulated symbol are not assumed to be independently impacted by noise. These bits are treated as a symbol in $\gfq$, so that a demodulated symbol is a element of $\gfq,q=2^m$ and, with a mild abuse of notation, we say that a demodulated vector is $\vym\in\gfq^{n/m}$, the most likely demodulated version of $\vy$ over $\gfq^{n/m}$.

In the proof of Theorem~\ref{thm:sgrand_correctness} we argued that at each iteration Algorithm~\ref{alg:sgrandab} removes the queried error sequence $\ve$ from $\cS$, performs code book membership test, and adds all the children of $\ve$ to $\cS$. The proof relies on error sequences satisfying Lemma~\ref{lem:predecessor_properties}, and not on the specific definition of parents and children. As a result, any definition of parents and children that satisfies Lemma~\ref{lem:predecessor_properties} would result in an algorithm that satisfies Theorem~\ref{thm:sgrand_correctness}.


For the $i$-th demodulated symbol, we define the \textit{ordered symbol indices} (OSI) vector $\vs^i\in\gfq^q$ as a vector that satisfies the following: 1) $s_1^i=\vym_i$, 2) $p\paren{y_{\govern{i}}\mid s_j^i}\geq p\paren{y_{\govern{i}}\mid s_l^i}\geq $ for $j<l$, 3) $\vs^i$ contains each element of $\gfq$ once. An OSI vector depends on $\vy$, the modulation and the statistics of the channel, but this dependency is made implicit for simplicity. We assume that a priori demodulated distribution is uniform over $\gfq$. An OSI vector of the $i$-th symbol is a vector that rank orders all possible demodulated symbols from most likely to least likely. For the $i$-th symbol we define $\prev{s_j^i}=s_{j-1}^i$, and $\prev{s_1^i}$ is not defined. For example, in the binary case for every bit we have $s_1^i=\vym_i,s_2^i=1-\vym_i$, $\prev{1-\vym_i}=\vym_i$, where the subtraction is in $\gftwo$.

Let $\ve\in\gfq^{n/m}$ be a non-zero error vector, $\vi$ be an OEI vector, $\jstar$ be defined as in Algorithm~\ref{alg:sgrandab}, and $\vs^{\jstar}$ be an OSI of the $\jstar$-th symbol. In this scenario we define the \textit{parent} of $\ve$, denoted by $\pred{\ve}$, as the following vector:
\[
    {\paren{\pred{\ve}}}_{i_j} = \left\{\begin{array}{lr}
        \prev{e_{i_j}}, & \text{if } j=\jstar, \\
        e_{i_j}, & otherwise
        \end{array}\right.
\]
while maintaining the same notation of $\ve$ being a child of $\pred{\ve}$. Note that by using this definition, each error has at most $n/m$ children. For example, when $n/m=4,q=4$, and assuming the OSI vector of each symbol is $\paren{0,1,\alpha,\beta}$, for an OEI vector $\vi=\paren{1,2,3,4}$, we have $\paren{1,0,0,0}=\pred{\paren{\alpha,0,0,0}}=\pred{\paren{1,1,0,0}}=\pred{\paren{1,0,1,0}}=\pred{\paren{1,0,0,1}}$. Lemma~\ref{lem:predecessor_properties} still holds, with the exception of Property~\ref{lemprop:pred_chain} requiring at most $\paren{q-1}\jstar$ compositions. The proof is similar to the proof given in~\ref{sub:proof_thm}, and is not given to avoid duplicity. Notice that if we try to define a parent in a similar fashion to the definition of the previous case, i.e. try to change the value of the $\paren{i_{\jstar-1}}$-th symbol, we may violate Property~\ref{lemprop:pred_prob} of Lemma~\ref{lem:predecessor_properties}.

SGRNADAB can be implemented using the new definition of a parent and children, where at each iteration the algorithm adds to $\cS$ all the children of the queried error sequence. Theorem~\ref{thm:sgrand_correctness} holds in this case, and the proof is similar to the original proof, since Lemma~\ref{lem:predecessor_properties} holds. For similar reasons as before, the worst-case complexity in this scenario is $\bigo\paren{b\cdot f\paren{n,k}\log\paren{bn/m}}$, as at each iteration at most $n/m$ elements added $\cS$, while one is removed from $\cS$. Thus, unlike most decoding algorithms, SGRAND can be adapted to manage the joint corruption of bits within a symbol by noise.

\section{Performance Evaluation}\label{sec:simulations}

Polar codes were the first provably capacity-achieving non-random code construction~\cite{arikan2008channel,arikan2009channel}. While they hold  promise for reliable communication, poor BLER results were reported based on their initial decoders~\cite{tal2011list,tal2015list,pfister2014brief}. Consequently, for all 5G NR control channel communications, a Cyclic Redundancy Check (CRC) is first appended and the whole result Polar coded, resulting in CA-Polar codes. 

Existing competitive decoding algorithms use this concatenated structure by first generating a list of possible candidates from the Polar code, and then selecting the best candidate that satisfies the CRC. When combined with soft information, the resulting CA-SCL decoders have shown considerable improvement in BLER performance~\cite{tal2011list,tal2015list, niu2012crc, balatsoukas2015llr}. As a state-of-the-art implementation of the CA-SCL decoding technique is available in Matlab's 5G toolbox version R2019a~\cite{matlab5g}, we compared SGRANDAB's BLER performance in that environment, relying mostly on Matlab's own code, see Fig.~\ref{fig:simulation}.

Note that considered as a single code, CA-Polar codes are themselves linear. They are first encoded with a CRC, interleaved if Downlink, and then encoded with a Polar code. All these operations are linear, hence the encoded message is the result of a linear encoding, $\vc=\vu\, G_{CRC}\, M_{interleave}\, G_{Polar}$ where $G_{CRC}$, $G_{Polar}$ are the generator matrices of the CRC code and the Polar code respectively, and $M_{interleave}$ is the interleaving matrix, which is the identity for Uplink communications. Thus from SGRANDAB's point of view, this is a single linear code and Algorithm~\ref{alg:sgrandab} can be used for its direct decoding of the concatenated code.

We compared SGRANDAB's performance with Matlab's CA-SCL decoder. While Matlab's default list size for CA-SCL is $L=8$~\cite{matlab5g}, also suggested in~\cite{3gppcapolar},\cite{xiang2019crc}, which is often considered the baseline for 5G performance evaluations~\cite{bioglio2018design}, at the cost of increased decoding times, we observed enhanced BLER performance with larger list sizes and so report results with $L=32$, as done for example in~\cite{tal2011list},\cite{niu2012crc},\cite{sarkis2015fast},\cite{wang2016parity}. The simulation was conducted on AWGN channel where Guassian noise with 0 mean and spectral density $\sigma^2=N_0/2$ is added independently to each modulated symbol. Here Signal to Noise Ratio (SNR) is defined as $\SNR=-10\log_{10}\paren{\sigma^2}\,\sbrace{dB}$, and $\ebno=\SNR-10\log_{10}\paren{k/n}-10\log_{10}\paren{m}\,\sbrace{dB}$, though we note that in Matlab's demo $\ebno=\SNR-10\log_{10}\paren{(k+\crclen)/n}-10\log_{10}\paren{m}$, where $\crclen$ is the number of CRC bits, corresponding to the Polar code alone\footnote{Matlab's technical support indicates this is to be consistent with 3GPP practice.}. 

Setting an abandonment threshold of $b=10^6$ queries for SGRANDAB, whereupon an erasure is reported, a BLER comparison for the [128,105] Uplink code, subject to Binary Phase Shift Keying (BPSK) modulation, is given in Fig.~\ref{fig:bler}. With the same abandonment threshold, an equivalent result for the [64,46] Uplink code, subject to Quadrature Phase Shift Keying (QPSK) modulation, is shown in Fig.~\ref{fig:simulation_bler}. For the longer code, SGRANDAB outperforms CA-SCL by $\approx 0.5\,\db$, and for the shorter code by $\approx 1\,\db$, indicating that further performance is available from these codes. In Fig.~\ref{fig:simulation_bler}, for SGRANDAB errors are not dominated by erasures, so the BLER curve is a good approximation of an unconstrained ML decoder's BLER curve. 

As a proxy for computational complexity, Fig.~\ref{fig:guesswork} and Fig.~\ref{fig:simulation_guesswork} report box plots for the number queries performed until a decoding is found. In all cases, this is significantly lower than the threshold $b$, which directly affects the complexity of the decoder, as discussed in~\ref{sub:complexity}. In particular, SGRANDAB's complexity gets better as the SNR improves. For comparison, a na\"ive ML decoder that computes the likelihood of each of the $2^k$ possible code words would require $\approx 10^{13}$, and $\approx 10^{31}$ likelihood calculations for the [64,46] and [128,105] codes, respectively. We also observe that in an AWGN channel, assuming bits are well-interleaved and each bit is independently impacted by a white Gaussian noise, the querying order does not depend on the SNR. The reason for that is that SNR affects the noise's variance, that would result in multiplying the log likelihood ratio of each bit by the same constant, which does not affect the querying order. In other words, assuming an AWGN channel corrupts transmitted code words, SGRANDAB does not require knowledge of the SNR in order to determine the querying order.

\begin{figure}
    \centering
    \includegraphics[width=1 \columnwidth]{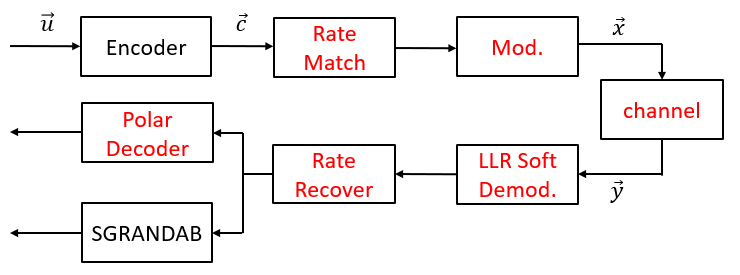}
    \caption{Simulation overview. Blocks containing red text are used as implemented in Matlab's 5G toolbox.}
    \label{fig:simulation}
\end{figure}

\begin{figure}
    \centering
    \includegraphics[width=1 \columnwidth]{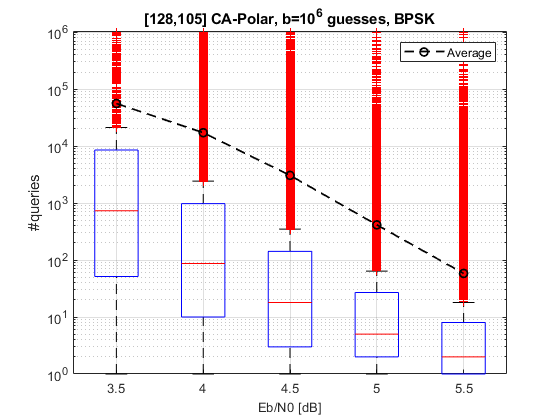}
    \caption{Number of queries.}
    \label{fig:guesswork}
\end{figure}

\begin{figure}
\centering
\subfloat[Block Error Rate.]{\includegraphics[width=1 \columnwidth]{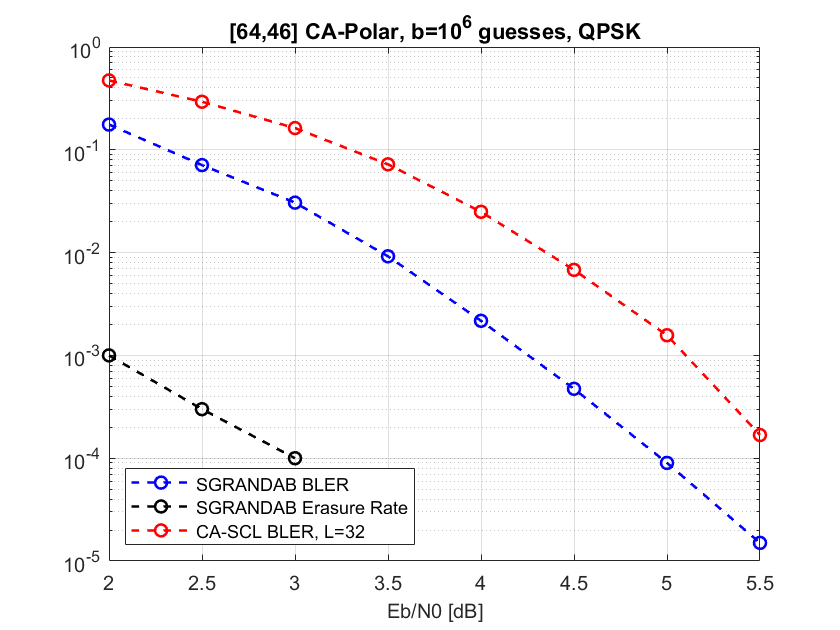}\label{fig:simulation_bler}}%
\linebreak
\subfloat[Number of queries]{\includegraphics[width=1 \columnwidth]{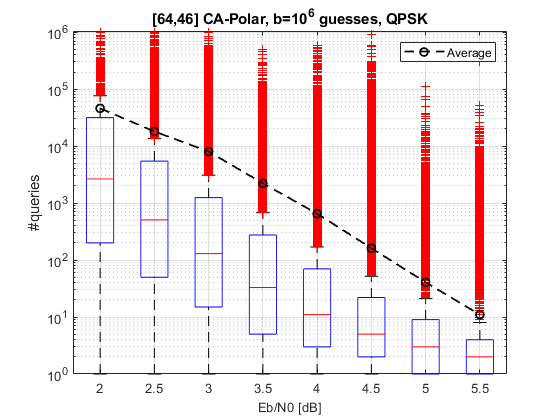}%
\label{fig:simulation_guesswork}}
\caption[]{Block Error Rate (BLER) comparison of Soft Guessing Random Additive Noise Decoder with Abandonment (SGRANDAB) and Matlab's implementation of~\cite{tal2011list,tal2015list} in Additive White Gaussian Noise subjet to QPSK modulation. $b$ is the abandonment threshold, as defined in Algorithm~\ref{alg:sgrandab}. For further details, see Section~\ref{sec:simulations}.}
\label{fig:comparison}
\end{figure}

\section{Conclusion and Discussion}\label{sec:conclusion}
In this paper we introduced SGRAND, a new soft detection ML decoder that is suitable for use with any short-length high-rate block code. For well interleaved channels, we formally proved that when SGRAND returns a code word, it is a maximum likelihood code word and hence no other decoder can be more accurate. We investigated its performance on 5G NR CA-Polar codes, which will be used for control channel communications in 5G NR, and compared its performance with a state-of-the-art soft CA-SCL decoder, establishing that an extra $0.5-1\,\db$ gain is possible through SGRANDAB's use.

While we have introduced SGRAND for memoryless channels, further generalisations of the approach are possible, but are not included here due to space constraints. One variant of SGRAND that can readily be developed is for use with a Markovian channel model, assuming the state of the channel is known.


\ifCLASSOPTIONcaptionsoff
  \newpage
\fi

\bibliographystyle{IEEEtran}
\bibliography{./bibtex/bib}

\end{document}